\documentclass[journal,10pt,twocolumn,twoside]{IEEEtran}
\usepackage{amsmath,amsfonts,amssymb,amsbsy,bm,paralist,theorem,color}
\usepackage{graphicx,algorithmic,algorithm,relsize}
\usepackage{multicol}
\usepackage{multirow}
\usepackage{subfigure}
\usepackage{caption}
\usepackage{cite}

\newtheorem{lem}{Lemma}

\graphicspath{{fig/}}

\definecolor{orange}{RGB}{255,107,0}
\definecolor{green}{RGB}{0,160,20}

\begin{document}
\title{Joint Optimization of Beamforming, Phase-Shifting and Power Allocation in a Multi-cluster IRS-NOMA Network}
\author{Ximing Xie, Fang Fang,~\IEEEmembership{Member,~IEEE}, and Zhiguo Ding,~\IEEEmembership{Fellow,~IEEE}
\thanks{Ximing Xie and Zhiguo Ding are with School of Electrical and Electronic Engineering, The University of Manchester, M13 9PL, U.K. (e-mail: ximing.xie@manchester.ac.uk and zhiguo.ding@manchester.ac.uk).}
\thanks{Fang Fang  is with the Department of Engineering, Durham University, Durham DH1 3LE, U.K. (e-mail:fang.fang@durham.ac.uk).}
}\maketitle

\begin{abstract}
The combination of non-orthogonal multiple access (NOMA) and intelligent reflecting surface (IRS) is an efficient solution to significantly enhance the energy efficiency of the wireless communication system. In this paper, we focus on a downlink multi-cluster NOMA network, where each cluster is supported by one IRS. We aim to minimize the transmit power by jointly optimizing the beamforming, the power allocation and the phase shift of each IRS. The formulated problem is non-convex and challenging to solve due to the coupled variables, i.e., the beamforming vector, the power allocation coefficient and the phase shift matrix. To address this non-convex problem, we propose an alternating optimization based algorithm. Specifically, we divide the primal problem into the two subproblems for beamforming optimization and phase shifting feasiblity, where the two subproblems are solved iteratively. Moreover, to guarantee the feasibility of the beamforming optimization problem, an iterative algorithm is proposed to search the feasible initial points. To reduce the complexity, we also propose a simplified algorithm based on partial exhaustive search for this system model. Simulation results demonstrate that the proposed alternating algorithm can yield a better performance gain than the partial exhaustive search algorithm, OMA-IRS, and NOMA with random IRS phase shift.
\end{abstract}
\vspace{-0.8em}
\section{Introduction}
The 5G communication system has been commercialized world-widely, and the beyond 5G (B5G) system starts attracting more and more researchers' attention due to its low energy consumption, high spectrum efficiency and massive multi-device interconnections \cite{david20186g,saad2019vision,tariq2019speculative}. In order to satisfy the increasing demand caused by the fast-growing number of users, various techniques, including millimetre wave \cite{jamali2019intelligent}, massive multi-inputs and multi-outputs (MIMO) system \cite{larsson2014massive}, and small cell \cite{guo2016method}, have been investigated and extensively used in practice. As a potential technique of B5G, non-orthogonal multiple access (NOMA) has received widespread attention due to its high spectral efficiency \cite{ding2015application, dai2018survey}. Different from conventional orthogonal multiply access (OMA), such as frequency division multiple access (FDMA), time division multiple access (TDMA), code division multiple access (CDMA), and orthogonal frequency-division multiple access (OFDMA), NOMA allows multiple users to share the same time slot, frequency block and channel code, which dramatically increases the spectral efficiency. In particular, the users in a NOMA network usually adopt successive inference cancellation (SIC) to remove the inference from other NOMA users, which can efficiently improve the signal to interference and noise ratio (SINR) and reception reliability \cite{saito2015performance}. Recently, intelligent reflective surface (IRS) has also been proposed as a potential solution to further improve the performance of wireless networks, including enlarging the communication coverage, and improving transmission robustness. Specifically, the IRS can reflect the electromagnetic wave to extend the cover rage of the base station (BS). It also has the ability to tune the channel by adjusting the phase shift of each element, which will greatly improve the quality of users' received signal\cite{wu2018intelligent}. \par 
\subsection{Related Works}
In literature, extensive research has been carried out for the NOMA technique, which has been combined with various state-of-the-art techniques such as MIMO and orthogonal time-frequency space modulation (OTFS)\cite{sun2018joint, fang2017joint, surabhi2019diversity, tang2019wireless, fang2016energy}. Recently, IRS has emerged as a kind of powerful equipment for wireless communication networks \cite{wu2019towards, di2019smart, wu2019intelligent}. Among these works, IRS was proved as a perfect solution for a wireless communication network, where the channel will be intelligently reconfigured by the IRS\cite{tang2019wireless, tang2019wireless1, ASIRSNOMAMag, FangIRS2020}. \par

Motivated by the benefits from NOMA and IRS, the combination of NOMA and IRS has been recently proposed as a promising solution to improve the communication systems. There have been some ongoing works studying the combination of NOMA and IRS. Some recent research works such as \cite{zhu2019power, ding2020impact} considered a simple scenario where a single IRS serves two users in a downlink NOMA network. In \cite{zhu2019power}, the authors minimized the transmit power at the BS by optimizing beanforming and IRS phase shift and also considered an improved quasi-degradation condition to guarantee that NOMA can achieve the capacity region with high possibility. In \cite{ding2020impact}, the authors analysed two kinds of phase shift designs, namely random phase shifting and coherent phase shifting.

Moreover, there are many works considering an IRS-assisted NOMA network where a signal IRS serves multiple users \cite{fu2019reconfigurable, liu2020ris, mu2019exploiting, zuo2020resource, zeng2020sum}. The problems which have been researched can be divided into two categories, one is about the transmit power minimization \cite{fu2019reconfigurable, liu2020ris} and the other is about the the sum-rate maximization\cite{mu2019exploiting, zuo2020resource, zeng2020sum}. For the transmit power minimization problem, the authors in \cite{mu2019exploiting} minimized the total transmit power by optimizing beamforming vectors of each user and the phase shift design of the IRS in an IRS-empowered downlink NOMA network. \cite{liu2020ris} considered a single IRS assisted downlink NOMA network and adopted reinforcement learning to design the beamforming vectors which minimized the transmit power at the BS. Regarding to the sum rate maximization problem, the authors of \cite{mu2019exploiting} optimized the beamforming design to maximize the sum rate in a downlink MISO IRS aided NOMA system. \cite{zuo2020resource} discussed a multi-channel downlink communications IRS-NOMA framework, where the sum rate of multiple NOMA users served by one IRS was maximized by optimizing resource allocation to each user and jointly considering channel assignment and decoding order. \cite{zeng2020sum} considered an IRS-assisted uplink NOMA system where multiple NOMA users can only transmit data through an IRS to the BS. \par

There are also some works considering a multi-cluster system mode, i.e., users are divided into different clusters \cite{li2019joint, ni2020resource}. In \cite{li2019joint}, the authors discussed a downlink IRS-assisted NOMA network where two types of users named the central user and the cell edge user were assigned to different clusters. Each cluster had one central user, one cell edge user and one IRS serving all users. The authors minimized the transmit power at the BS by jointly optimizing the beamforming vectors of each user and the phase shift design of the IRS. In \cite{ni2020resource}, the authors considered a multi-cluster and multi-BS IRS-aided NOMA network, where each cluster is served by its associated BS and one IRS serves all clusters. The sum rate was minimized by jointly optimizing power allocation and phase shift.\par

\subsection{Motivation and Challenges}
All the above works only consider one IRS. However the channel state of each user is related to its particular surrounding environment. Therefore, one single IRS might not be enough to reconfigure all users' channels simultaneously. Thus, we propose the use of multiple IRSs to assist the users whose channel conditions are bad. One IRS can adjust its phase shift dedicatedly for its associated user to generate a better channel condition. In this paper, we consider a multi-cluster NOMA network, where each cluster has one IRS and the BS generates an unique beam for each cluster to serve all users located in this cluster. \par

With the considered scenario, there are a few challenges which need to be overcome. We consider a multi-user and multi-IRS scenario which increases the number of optimization variables and hence make the optimization more complicated than the case with a single IRS in the network. The joint optimization problem contains three coupled variables, which is a non-convex problem and highly intractable. We divide the primal problem into subproblems and transform them to convex forms through approximations, the feasibility of these subproblems cannot be guaranteed during the transformation. Moreover, due to the high quality of variables, the computing time of algorithms will be extensive.  

\subsection{Contributions}
Different from the above mentioned works \cite{li2019joint, ni2020resource}, in this paper, we adopt a new system model and employ multiple IRSs to assist users. Then, we formulate a non-convex optimization problems which is highly intractable. We propose a novel alternating algorithm to solve this non-convex problem efficiently. Finally, we simplify the system model and propose a low-complexity algorithm, which achieves a reasonable performance. We summarize the contributions as follows:
\begin{itemize}
\item We consider a multi-cluster IRS-NOMA system, where each cluster contains two users served by one IRS. We formulate the transmit power minimization problem with respect to the beamforming vector, the phase shifting matrix of IRSs and the power allocation coefficient of each cluster. Each IRS can accomplish channel reconfiguration according to the channel condition between the BS and the cell edge user it serves, which intuitively yields a better performance than the scenario with the single IRS serving the whole system.

\item The formulated problem is non-convex because three variables are highly coupled together. To solve the proposed optimization problem, we propose an alternating algorithm by decoupling variables. We divide the primal problem into two subproblems. However, the beamforming optimization problem still has two variables coupled together, which causes the intractability. To address this challenge, we first adopt arithmetic and geometric means inequalities to approximately transform the non-convex set to its upper bound which is convex. Then, we use the equivalence between Schur complement larger than zero and the positive semidefinite matrix and successive convex approximation (SCA) to transfer another non-convex constraint to a convex form. Finally, we use an alternating algorithm iteratively solve two subproblems.

\item We introduce some fixed points during the approximation. It is essential to obtain the initial choice of the fixed points to guarantee the feasibility of the beamforming optimization problem. Therefore, we propose a feasible initial points search algorithm, where we introduce an auxiliary variable to force all constraints to be feasible. We minimize this auxiliary variable until it equals to zero. The values of the fixed points when this auxiliary variable equals to zero can be the initial choice of the fixed points for the proposed alternating algorithm.

\item The complexity of the proposed alternating algorithm is high. To reduce the complexity of the proposed algorithm, we simplify the system model where each cluster shares the same power allocation coefficient. With this assumption, the previous problem will be degraded into a simpler one with two coupled variables. We design a partial exhaustive search algorithm to solve this new problem, which has a low complexity. Compared with the alternating algorithm, the complexity is reduced but the performance is still reasonable.
\end{itemize}
\subsection{Organization}
The rest of paper is organized as follows. In Section \uppercase\expandafter{\romannumeral3}, we describe a multi-cluster IRS-assisted NOMA downlink network and formulate a energy minimizing optimization problem. In Section \uppercase\expandafter{\romannumeral3}, the solution to solve the formulated problem is introduced. In Section \uppercase\expandafter{\romannumeral4}, we briefly illustrate the simplified optimization problems and the partial exhaustive search based algorithm. In Section \uppercase\expandafter{\romannumeral5}, we provide the convergence analysis of the algorithms. We also present the simulation results to analyze the performance of the proposed algorithm. Finally, a conclusion is summarised in Section \uppercase\expandafter{\romannumeral6}.

\section{System Model and Problem Formulation}
\subsection{System Model}
\begin{figure}
	\centering
	\graphicspath{{./figures/}}\includegraphics[width=0.8\linewidth]{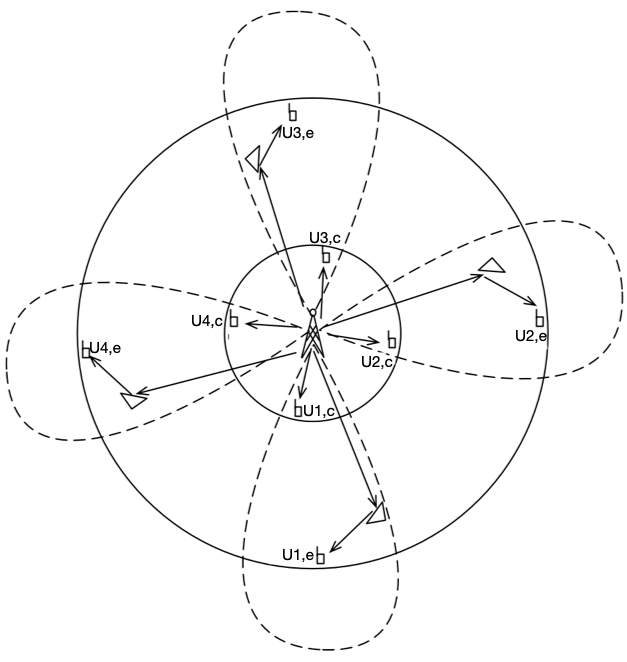}\\
	\caption{An IRS NOMA sytem model.} \label{Fig0}
\end{figure}
As shown in Fig.\ref{Fig0}, we consider a multi-user downlink network where two types of users, namely the central user and the cell edge user, are served by the BS simultaneously. Generally, the central users are much closer to the base station than the cell edge users. We assume that there are $K$ clusters and each cluster contains a central user, a cell edge user and an IRS. We use $CU_k$, $EU_k$ and ${\rm IRS}_k$ to represent the central user, the cell edge user and the IRS in the $k$-th cluster, respectively. Each IRS is equipped with $N$ passive reflecting elements and assists the cell edge user receiving signal from the BS. The BS is equipped with $M$ ($M \geq K$) antennas and $K$ beams to serve $K$ clusters. We assume that the direct links between the BS and all cell edge users are not available due to blockage, and the IRSs are implemented to reflect the signals sent by the BS to the cell edge users. Each cluster is far from others so the interference caused by the IRSs serving the other clusters can be reasonably ignored. In each cluster, the IRS is deployed closely to the cell edge user so that it will not infect the central user either. \par

To improve the spectrum efficiency, we adopt NOMA to serve all users simultaneously and we also assign different power levers to the two users in each cluster. The base station broadcasts the superposition signal $\sum_{k=1}^K \mathbf{w}_k(\alpha_k s_{k,c} + (1-\alpha_k)s_{k,e})$, where $\mathbf{w}_k \in \mathbb{C}^M$ denotes the beamforming vector in the $k$-th cluster and $k\in {1,2,...,K}$. $s_{k,c}$ and $s_{k,e}$ denote the signals to be sent to the central users and the cell edge users, respectively, and $\alpha_k$ is the power allocation coefficient of $CU_k$, thus $1-\alpha_k$ is the power coefficient of $EU_k$. We assume that the BS perfectly knows channel state information (CSI). Therefore, the signal received at $CU_k$ is given by 
\begin{align}
y_{k,c} = \mathbf{h}_{k,c}^H\sum_{k=1}^K \mathbf{w}_k(\alpha_k s_{k,c} + (1-\alpha_k)s_{k,e}) + w_{c,k},
\end{align}
where $\mathbf{h_{k,c} \in \mathbb{C}^{M\times 1}}$ denotes the channel vector between the base station and $CU_k$, and $w_{c,k} \sim \mathcal{CN}(0,\sigma^2)$ is the additive white Gaussian noise (AWGN). Meanwhile, the signal received at $CU_k$ is given by 
\begin{align}
y_{k,e} = (\mathbf{h}_{k,e}^H\mathbf{\Theta}_k\mathbf{G}_k)\sum_{k=1}^K \mathbf{w}_k(\alpha_k s_{k,c} + (1-\alpha_k)s_{k,e}) + w_{e,k},
\end{align}
where $\mathbf{G}_k \in \mathbb{C}^{N \times M}$ denotes the channel matrix between the BS and ${\rm IRS}_k$, $w_{e,k} \sim \mathcal{CN}(0,\sigma^2)$ denotes AWGN, $\mathbf{h}_{k,e} \in \mathbb{C}^{N\times 1}$ denotes the channel vector between ${\rm IRS}_k$ and $EU_k$, and $\mathbf{\Theta}_k = {\rm diag}(\beta e^{j\theta_1^k},...,\beta e^{j\theta_n^k})$ is the phase shift matrix of ${\rm IRS}_k$, where $\theta_n^k \in [0,2\pi), n \in \{1,...,N\}$ and $\beta \in [0,1]$ denote the phase shift of each reflecting element $n$ and amplitude coefficient on the signal, respectively. Without loss of generality, we assume $\beta=1$ given the fact that each reflecting element can only change the phase but not the amplitude of the reflected signal \cite{liu2020ris}. More detailed discussions about the choices of the reflecting amplitude and the phase shift can be found in \cite{abeywickrama2020intelligent}. Due to the path loss, we consider that the signal can be only efficiently reflected by the IRS once. Moreover, the long distance that geographically separates each cluster justifies the assumption that the IRS in one cluster will not infect other clusters. With this assumption, in each cluster, SIC is only performed at the central user to eliminate the interference from its intra-cluster edge user and the cell edge user decodes its data directly. Hence, the SINR of $EU_k$ is given by
\begin{align}
\mathrm{SINR}_{k,e}=\frac{|\mathbf{h}_{k,e}^H\mathbf{\Theta}_k\mathbf{G}_k\mathbf{w}_k|^2(1-\alpha_k)}{|\mathbf{h}_{k,e}^H\mathbf{\Theta}_k\mathbf{G}_k\mathbf{w}_k|^2\alpha_k+\sum\limits_{\substack{i=1 \\ i\neq k}}^K|\mathbf{h}_{k,e}^H\mathbf{\Theta}_k\mathbf{G}_k\mathbf{w}_i|^2+\sigma^2}, \label{eq3}
\end{align}
where $|\mathbf{h}_{k,e}^H\mathbf{\Theta}_k\mathbf{G}_k\mathbf{w}_k|^2\alpha_k$ is intra-cluster interference and $\sum_{\substack{i=1 \\ i\neq k}}^K|\mathbf{h}_{k,e}^H\mathbf{\Theta}_k\mathbf{G}_k\mathbf{w}_i|^2$ is inter-cluster interference. For the central users, they need to apply SIC to decode $s_{k,e}$ first and then remove it. Thus, the SINR of signal $s_{k,e}$ observed at $CU_k$ can be expressed as follows
\begin{align}
\mathrm{SINR}_{k,c\to e}=\frac{|\mathbf{h}_{k,c}^H\mathbf{w}_k|^2(1-\alpha_k)}{|\mathbf{h}_{k,c}^H\mathbf{w}_k|^2\alpha_k+\sum\limits_{\substack{i=1 \\ i\neq k}}^K|\mathbf{h}_{k,c}^H\mathbf{w}_i|^2+\sigma^2}.
\end{align}
The SNR of $CU_k$ to decode its own signal is given by
\begin{align}
\mathrm{SINR}_{k,c}=\frac{|\mathbf{h}_{k,c}^H\mathbf{w}_k|^2\alpha_k}{\sum\limits_{\substack{i=1 \\ i\neq k}}^K|\mathbf{h}_{k,c}^H\mathbf{w}_i|^2+\sigma^2}.
\end{align}

\subsection{Problem Formulation}
In this section, we formulate a transmit power minimization problem by jointly optimizing the beamforming vector $\left(\mathbf{w}_k, k \in \{1,...K\}\right)$, power allocation coefficients $\left(\alpha_k, k \in \{1,...K\}\right)$ and phase shifting matrix $\left(\mathbf{\Theta}_k, k \in \{1,...K\}\right)$, while considering the quality of service (QoS) requirement and the constraints of reflecting elements. The considered transmit power minimization problem can be formulated as

\begin{subequations}\label{Prob:0} 
\begin{align}
&{\rm P}0:\min_{\boldsymbol{\alpha},\mathbf{w},\mathbf{\Theta}}\quad\sum\limits_{k=1}^K||\mathbf{w}_k||^2 \\
		&~\mathrm{s.t.}~\log_2(1+\mathrm{SINR}_{k,c})\geq R_{k,c},\quad \forall k \label{P1b}\\
		&\qquad \log_2(1+ \min (\mathrm{SINR}_{k,e},\mathrm{SINR}_{k,c\to e}))\geq R_{k,e},\forall k \label{P1c}\\
		&\qquad 0\leq\theta_{i,n}\leq 2\pi,\quad \forall ~ i,n \label{P1d} \\
		&\qquad |\mathbf{\Theta}_{i,n,n}|\leq 1,\quad \forall ~ i,n \label{P1e}
\end{align}
\end{subequations}
where $||\mathbf{w}_k||^2$ is the transmit power allocated to the cluster $k$, $R_{k,c}$ and $R_{k,e}$ denote the required minimum data rate of $CU_k$ and $EU_k$, respectively. The constraints (\ref{P1b}) and (\ref{P1c}) indicate the QoS requirements of the central users and the cell edge users, (\ref{P1d}) defines the phase shift range of the reflecting elements and \eqref{P1e} ensures that the IRS is a passive component. \par
However, problem ${\rm P}0$ is highly intractable due to the non-convex constraints (\ref{P1b}) and (\ref{P1c}). The non-convexity is caused by three highly coupled variables (i.e. $\mathbf{w}$, $\alpha$ and $\mathbf{\Theta}$). To efficiently solve this problem, we adopt SCA, SDR and the inequality approximation to develop an alternating algorithm to iteratively solve it.
\section{ Optimization Solution}
As discussed in the previous section, it is difficult to find the optimal solution of P0 due to its non-convexity. In this section, an alternating optimization algorithm is proposed to solve P0 efficiently. The main idea of this algorithm is to divide the primal problem into two subproblems and solve them alternatively. In particular, P0 is divided to a beamforming optimization subproblem and a feasible phase shifting matrix search subproblem. As shown later, each of the two subproblems is non-convex, and we will transform them into convex forms which can be solved efficiently by convex solver, e.g., CVX in Matlab.

\subsection{Beamforming Optimization}
For a given phase shifting matrix $\mathbf{\Theta}$, the concatenated channel respond $\mathbf{h}_{k,e}^H\mathbf{\Theta}_k\mathbf{G}_k \in \mathbb{C}^{1 \times M}$ is fixed. Thus, the beamforming optimization problem can be formulated as
\begin{subequations}\label{Prob:1} 
\begin{align}
{\rm P}1:&\min\limits_{\mathbf{\alpha,\mathbf{w}}} \quad \sum\limits_{k=1}^K ||\mathbf{w}_k||^2 \\
		&{\rm s.t.} \quad\quad \log_2(1+{\rm SINR}_{k,c}) \geq R_{k,c}, \quad \forall k \label{P2a}\\
		& \qquad \;\quad \log_2(1+{\rm SINR}_{k,e}) \geq R_{k,e},\quad \forall k \label{P2b}\\
		& \qquad \;\quad \log_2(1+{\rm SINR}_{k,c \to e}) \geq R_{k,e},\quad \forall k \label{P2c} \\
		& \qquad \;\quad 0 \leq \alpha_k \leq 1, \quad \forall k. \label{P2d}
\end{align}
\end{subequations}
P1 is non-convex because the beamforming vector and the power allocation coefficient are still coupled together in all constraints except \eqref{P2d}, which is challenging to be solved. We notice that the rank-constrained semidefinite programming (SDP) problem can be approximated to a convex form. Therefore, we convert ${\rm P}1$ into a SDP form, then SDR can be applied to solve this problem. \par
First, we transform the constraint (\ref{P2b}) into a convex form. According to \eqref{eq3}, the constraint (\ref{P2b}) can be rewritten as follows:
\begin{align} 
\frac{|e_k^H \mathbf{D}_{k,e} \mathbf{G}_k \mathbf{w}_k|^2(1-\alpha_k)}{|e_k^H \mathbf{D}_{k,e} \mathbf{G}_k \mathbf{w}_k|^2\alpha+\sum\limits_{\substack{i=1 \\ i\neq k}}^K|e_k^H \mathbf{D}_{k,e} \mathbf{G}_k \mathbf{w}_i|^2 + \sigma^2}  \geq r_{k,e}, \label{eq8}
\end{align}
where $r_{k,e} = 2^{R_{k,e}} - 1$, $e_k$ is an $N \times 1$ vector containing all the diagonal elements of $\mathbf{\Theta}_k^H$, and $\mathbf{D}_{k,e}$ is a diagonal matrix, whose main diagonal elements are from the channel vector $\mathbf{h}_{k,e}^H$. After some algebraic transformations, (\ref{eq8}) can be equivalently expressed as follows:
\begin{equation} \label{eq9}
\begin{split}
&|e_k^H \mathbf{D}_{k,e} \mathbf{G}_k \mathbf{w}_k|^2(1+r_{k,e}) \alpha_k \leq \\
&|e_k^H \mathbf{D}_{k,e} \mathbf{G}_k \mathbf{w}_k|^2 - \sum\limits_{\substack{i=1 \\ i\neq k}}^K|e_k^H \mathbf{D}_{k,e} \mathbf{G}_k \mathbf{w}_i|^2r_{k,e} - \sigma^2r_{k,e}.
\end{split}
\end{equation}
Since the CSI is perfectly known by the BS, the channel $e_k^H \mathbf{D}_{k,e} \mathbf{G}_k$ is fixed with a given phase shifting matrix. For simply notation, we replace $e_k^H \mathbf{D}_{k,e} \mathbf{G}_k$ with $\mathbf{z}_{k,e}^H$ and rewrite (\ref{eq9}) as follows:
\begin{equation}
\alpha_k|\mathbf{z}_{k,e}^H \mathbf{w}_k|^2 \leq \frac{|\mathbf{z}_{k,e}^H \mathbf{w}_k|^2}{1+r_{k,e}} - (\sum\limits_{\substack{i=1 \\ i\neq k}}^K |\mathbf{z}_{k,e}^H \mathbf{w}_i|^2 + \sigma^2) \frac{r_{k,e}}{1+r_{k.e}}, \label{eq10}
\end{equation}
where $\mathbf{z}_{k,e}^H = e_k^H \mathbf{D}_{k,e} \mathbf{G}_k$. Note that the beamforming vector in (\ref{eq10}) has the same form as $\mathbf{w}_k \mathbf{w}_k^H$. Inspired by SDR, we introduce a slack matrix $\mathbf{W}_k = \mathbf{w}_k \mathbf{w}_k^H$, which is a rank-one positive semidefinite (PSD) matrix. Then the constraint (\ref{eq10}) can be equivalently rewritten as follows:
\begin{align}
&\alpha_k \rm{Tr}(\mathbf{Z}_{k,e} \mathbf{W}_k) \leq \notag \\
&\frac{\rm{Tr}(\mathbf{Z}_{k,e}\mathbf{W}_k)}{1+r_{k,e}} - (\sum\limits_{\substack{i=1 \\ i\neq k}}^K {\rm Tr}(\mathbf{Z}_{k,e}\mathbf{W}_i) + \sigma^2) \frac{r_{k,e}}{1+r_{k.e}} \label{eq11} \\
&\mathbf{W}_k \succcurlyeq 0 \label{eq12} \\
&\rm{Rank}(\mathbf{W}_k) = 1, \label{eq13}
\end{align}
where $\mathbf{Z}_{k,e} = \mathbf{z}_{k,e} \mathbf{z}_{k,e}^H$. From (\ref{eq11}), we notice that the right hand side of (\ref{eq11}) is a liner combination of two convex terms with respect to $\mathbf{W}_k$, which is convex. The only obstacle is the left hand side, which is a bilinear term constructed by  $\alpha_k$ and $\rm{Tr}(\mathbf{Z}_{k,e} \mathbf{W}_k)$. To make this constraint a convex set, we need to transform the non-convexity function $\alpha_k \rm{Tr}(\mathbf{Z}_{k,e} \mathbf{W}_k)$ to a convex form. Inspired by the inequality of arithmetic and geometric means, the non-convex feasible set of the left hand side term can be upper bounded by a convex set $\frac{1}{2} (\alpha_k^2 + \rm{Tr}(\mathbf{Z}_{k,e}\mathbf{W}_k)^2)$. To tighten this upper bound in each iteration of the proposed iterative algorithm, we introduce a fixed point $c_k$, then we have 
\begin{align}
2\alpha_k {\rm Tr}(\mathbf{Z}_{k,e} \mathbf{W}_k) \leq (\alpha_k c_k)^2 + \left(\frac{({\rm Tr}(\mathbf{Z}_{k,e} \mathbf{W}_k)}{c_k}\right)^2. \label{eq14}
\end{align}
We iteratively update this fixed feasible point.
\begin{lem} \label{lemma1}
The fixed point $c_k$ at the $m$-th iteration can be updated by:
\begin{align}
c_k^{(m)} = \sqrt{\frac{{\rm Tr}(\mathbf{Z}_{k,e}\mathbf{W}_k^{(m-1)})}{\alpha_k^{(m-1)}}} \label{eq15}
\end{align}
\end{lem}
\begin{proof}
We define the difference function of the original function $2\alpha_k {\rm Tr}(\mathbf{Z}_{k,e} \mathbf{W}_k)$  and its approximated upper bound as
\begin{align}
\mathcal{F}(c_k) = 2\alpha_k {\rm Tr}(\mathbf{Z}_{k,e} \mathbf{W}_k) - (\alpha_k c_k)^2 - \left(\frac{({\rm Tr}(\mathbf{Z}_{k,e} \mathbf{W}_k)}{c_k}\right)^2 \label{eq16}
\end{align}
when the function \eqref{eq16} equals to 0, both sides of \eqref{eq14} are equal, which tightens the upper bound. From \eqref{eq14}, we notice that the maximum value of function $\mathcal{F}(c_k)$ is 0. Since 
\begin{align}
\frac{\partial^2\mathcal{F}(c_k)}{\partial c_k^2} = -2 \alpha_k - \frac{6 {\rm Tr}(\mathbf{Z}_{k,e} \mathbf{W}_k)}{c_k^4} \leq 0,
\end{align}
when $\alpha_k \geq 0$ and ${\rm Tr}(\mathbf{Z}_{k,e} \mathbf{W}_k) \geq 0$, the function $\mathcal{F}(c_k)$ is a concave function with respect to $c_k$. The optimal value of $c_k$, defined as $c_k^*$, can be obtained by $\frac{\partial^2\mathcal{F}(c_k)}{\partial c_k^2} = 0$, then we have
\begin{align}
c_k^{*} = \sqrt{\frac{{\rm Tr}(\mathbf{Z}_{k,e}\mathbf{W}_k)}{\alpha_k}}. \label{up_c}
\end{align} 
\end{proof}
Hence, the constraint \eqref{P2b} can be approximated as follows:
\begin{equation} \label{eq19}
\begin{split}
&(\alpha_k c_k)^2 + \left(\frac{({\rm Tr}(\mathbf{Z}_{k,e} \mathbf{W}_k)}{c_k}\right)^2 \leq \\ 
&2\frac{{\rm Tr}(\mathbf{Z}_{k,e}\mathbf{W}_k)}{1+r_{k,e}} - 2(\sum\limits_{\substack{i=1 \\ i\neq k}}^K {\rm Tr}(\mathbf{Z}_{k,e}\mathbf{W}_i) + \sigma^2) \frac{r_{k,e}}{1+r_{k.e}}. \\
\end{split}
\end{equation}
It is noted that the left hand side of \eqref{eq19} is convex and the right side of \eqref{eq19} is an affine function, which means that the constraint \eqref{eq19} is a convex set.  \par
For handling with the next non-convex constraint \eqref{P2c}, after some algebraic manipulations, we can rewrite \eqref{P2c} as follows:
\begin{align}
\alpha_k|\mathbf{h}_{k,c}^H \mathbf{w}_k|^2 \leq \frac{|\mathbf{h}_{k,c}^H \mathbf{w}_k|^2}{1+r_{k,e}} - (\sum\limits_{\substack{i=1 \\ i\neq k}}^K |\mathbf{h}_{k,c}^H \mathbf{w}_i|^2 + \sigma^2) \frac{r_{k,e}}{1+r_{k.e}}. \label{eq20}
\end{align}
It is worth to point out that \eqref{eq20} has the same form as \eqref{eq10}. Similarly, the method allied to \eqref{eq10} can be efficiently applied to \eqref{eq20} to yield a convex form. Therefore, \eqref{eq20} can be eventually transformed to
\begin{equation} \label{eq21}
\begin{split}
&(\alpha_k d_k)^2 + \left(\frac{({\rm Tr}(\mathbf{H}_{k,c} \mathbf{W}_k)}{d_k}\right)^2 \leq \\ 
&2\frac{{\rm Tr}(\mathbf{H}_{k,c}\mathbf{W}_k)}{1+r_{k,e}} - 2(\sum\limits_{\substack{i=1 \\ i\neq k}}^K {\rm Tr}(\mathbf{H}_{k,c}\mathbf{W}_i) + \sigma^2) \frac{r_{k,e}}{1+r_{k.e}} \\
\end{split}
\end{equation}
where $\mathbf{H}_{k,c} = \mathbf{h}_{k,c} \mathbf{h}_{k,c}^H$, and $d_k$ is a fixed point. At the $m$-th iteration, $d_k$ can be updated as follows:
\begin{align}
d_k^{(m)} = \sqrt{\frac{{\rm Tr}(\mathbf{H}_{k,c}\mathbf{W}_k^{(m-1)})}{\alpha_k^{(m-1)}}}. \label{up_d}
\end{align} \par
Now, we focus on the last non-convex constraint \eqref{P2a}. First, we also rewrite it as follows:
\begin{align}
\alpha_k {\rm Tr}(\mathbf{H}_{k,c} \mathbf{W}_k) \geq \sum\limits_{\substack{i=1 \\ i \neq k}}^K {\rm Tr}(\mathbf{H}_{k,c} \mathbf{W}_i) r_{k,c} + \sigma^2 r_{k,c} \label{eq23}
\end{align}
where $r_{k,c} = 2^{R_{k,c}} - 1$. Though \eqref{eq23} also has a bilinear term $\alpha_k {\rm Tr}(\mathbf{H}_{k,c} \mathbf{W}_k)$, we cannot straightforwardly apply the method which has been successfully applied to constraint \eqref{P2b} and \eqref{P2c}. Even we replace $\alpha_k {\rm Tr}(\mathbf{H}_{k,c} \mathbf{W}_k)$ with the sum of two square terms through the inequality of arithmetic and geometric means they are located at the left side of $\geq$ sign, which causes this inequality to be concave, and the transformed constraint is still non-convex. Hence, we propose another method to deal with this constraint. First, we introduce a slack variable $t_k$ and \eqref{P2a} can be equivalently transformed to 
\begin{align}
\alpha_k {\rm Tr}(\mathbf{H}_{k,c} \mathbf{W}_k) \geq t_k^2 \label{eq24}
\end{align}
\begin{align}
t_k^2 \geq \sum\limits_{\substack{i=1 \\ i \neq k}}^K {\rm Tr}(\mathbf{H}_{k,c} \mathbf{W}_i)r_{k,c} + \sigma^2r_{k,c}. \label{eq25}
\end{align}
It can be straightforwardly show that neither of \eqref{eq24} and \eqref{eq25} is convex. According to the convex optimization theory \cite{boyd2004convex}, we know that the sufficient and necessary condition for a matrix to be PSD is that its Schur complement is greater than zero and also know that a PSD matrix is a convex constraint. After a simple transformation, \eqref{eq24} can be rewritten as follows:
\begin{align}
\alpha_k - \frac{t_k^2}{{\rm Tr}(\mathbf{H}_{k,c} \mathbf{W}_k)} \geq 0, \label{eq26}
\end{align}
which is equivalent to
\begin{align}
\begin{bmatrix}
		\alpha_k & t_k \\
		t_k & {\rm Tr}(\mathbf{H}_{k,c} \mathbf{W}_i)
		\end{bmatrix} \succcurlyeq 0. \label{eq27}
\end{align}
Constraints \eqref{eq26} and \eqref{eq27} are mutually sufficient, and constraint \eqref{eq27} is convex. Now, we deal with constraint \eqref{eq25}. We notice that $t_k^2$ is on the left hand side of the greater sign, which makes the whole constraint a non-convex set. To deal with this, we adopt SCA, where the first order Taylor series approximation is adopted to approximate the quadratic form \eqref{eq25} to
\begin{align}
t_{k,0}^2 + 2t_{k,0}(t_k - t_{k,0}) \geq \sum\limits_{\substack{i=1 \\ i \neq k}}^K {\rm Tr}(\mathbf{H}_{k,c} \mathbf{W}_i)r_{k,c} + \sigma^2r_{k,c}
\end{align}
where $t_{k,0}$ is a fixed point. By applying SCA, we update $t_{k,0}$ at the $m$-th iteration by
\begin{align}
t_{k,0}^{(m)} = t_k^{(m-1)}. \label{up_t}
\end{align}
The final obstacle to deal with this problem arises from the rank-one constraint \eqref{eq13}. By applying SDR, the rank-one constraint is omitted to make the whole problem tractable. Thus we eventually transform P1 to
\begin{subequations}\label{Prob:2} 
\begin{align}
&{\rm P2}:\min\limits_{\mathbf{\alpha,\mathbf{w}, \mathbf{t}}} \quad \sum\limits_{k=1}^K {\rm Tr}(\mathbf{W}_k) \\
		&{\rm s.t.} \quad (\alpha_k c_k)^2 + \left(\frac{({\rm Tr}(\mathbf{Z}_{k,e} \mathbf{W}_k)}{c_k}\right)^2 \leq \notag\\
		 &\quad  2\frac{{\rm Tr}(\mathbf{Z}_{k,e}\mathbf{W}_k)}{1+r_{k,e}} - 2(\sum\limits_{\substack{i=1 \\ i\neq k}}^K {\rm Tr}(\mathbf{Z}_{k,e}\mathbf{W}_i) + \sigma^2) \frac{r_{k,e}}{1+r_{k.e}}, \forall k\label{P3a}\\
		& \qquad  (\alpha_k d_k)^2 + \left(\frac{({\rm Tr}(\mathbf{H}_{k,c} \mathbf{W}_k)}{d_k}\right)^2 \leq \notag \\
		&\quad 2\frac{{\rm Tr}(\mathbf{H}_{k,c}\mathbf{W}_k)}{1+r_{k,e}} - 2(\sum\limits_{\substack{i=1 \\ i\neq k}}^K {\rm Tr}(\mathbf{H}_{k,c}\mathbf{W}_i) + \sigma^2) \frac{r_{k,e}}{1+r_{k.e}}, \forall k\label{P3b}\\
		& \quad  \begin{bmatrix}
		\alpha_k & t_k \\
		t_k & {\rm Tr}(\mathbf{H}_{k,c} \mathbf{W}_i)
		\end{bmatrix} \succcurlyeq 0, \forall k\label{P3c}\\
		& \quad  t_{k,0}^2 + 2t_{k,0}(t_k - t_{k,0}) \geq \sum\limits_{\substack{i=1 \\ i \neq k}}^K {\rm Tr}(\mathbf{H}_{k,c} \mathbf{W}_i)r_{k,c} + \sigma^2r_{k,c},\forall k \label{P3d} \\
		& \quad 0 \leq \alpha_k \leq 1, \forall k. \label{P3e}
\end{align}
\end{subequations}
\begin{algorithm}[tp]
	\caption{Initial Point Search Algorithm }\label{Alg1}
	\begin{algorithmic}[1] 
		\STATE {{\bf Initialize:} $c_k^{(0)}$, $d_k^{(0)}$, $t_{k,0}^{(0)} \quad \forall k$, $\epsilon = 0.00001$, $i=0$, $q^{(0)} = 100$.}
		\WHILE {$q^{(i)} > \epsilon$}
		\STATE {$i = i+1.$}
		\STATE {Update $\mathbf{W}_k^{(i)}$, $\alpha_k^{(i)}$ and $q^{(i)}$ with fixed $c_k^{(i-1)}$, $d_k^{(i-1)}$, $t_{k,0}^{(i-1)}$ by solving P3.}
		\STATE {Update $c_k^{(i)}$, $d_k^{(i)}$, and $t_{k,0}^{(i)}$ based on \eqref{up_c}, \eqref{up_d} and \eqref{up_t} respectively.}
		\ENDWHILE
		\STATE {{\bf Output} $c_k^{(i)}$, $d_k^{(i)}$, and $t_{k,0}^{(i)}.$}
	\end{algorithmic}\label{Al1}
\end{algorithm}

Since the restriction of rank one is removed, P2 is a convex problem and can be efficiently solved by convex optimization toolboxes, for instance, CVX. However, the optimal solution of P2 may not be the optimal solution of P1 unless the rank of $\mathbf{W}_k^*, \forall k$ is 1. We use Gaussian randomization \cite{luo2010semidefinite} to alternatively obtain a suboptimal solution. We define the optimal solution of P2 as $\mathbf{W}_k^*, \forall k$, and each $\mathbf{W}_k^*$ is a positive semidefinite matrix. Recall that $\mathbf{W}_k^* = \mathbf{w}_k^* \mathbf{w}_k^{*H}$, eigenvalue decomposition can be used to obtain the optimal beamforming vector.  \par

Before we solve P2, we need to initialize three fixed points, $c_k$, $d_k$ and $t_{k,0}, \forall k$. It is noted that initializing them randomly will make the formulated problem infeasible. Hence, we propose a feasible initial points search algorithm to find the feasible fixed points to make P2 solvable. From the P2, we notice that the fixed points $c_k$. $d_k$ and $t_{k,0}$ must satisfy the constraints \eqref{P3a}, \eqref{P3b} and \eqref{P3c}. To address this problem, we introduce an auxiliary variable $q$, which intentionally relaxes the constraints to enlarge the feasible set. Then, we can formulate the initial point search problem as
\begin{subequations}\label{Prob:3} 
\begin{align}
&{\rm P}3:\min\limits_{\mathbf{\alpha,\mathbf{w}},\mathbf{t},q} \quad q \\
		&{\rm s.t.} \quad \eqref{P3a} \; \eqref{P3b} \; \eqref{P3c} \; \eqref{P3d} \; \eqref{P3e}  \\
		& \quad \; \quad q \geq 0. \label{P4f}
\end{align}
\end{subequations}
Specifically, when $q$ equals to 0, all constraints in P3 are exactly the same as the constraints in P2 and the obtained values of $c_k$, $d_k$ and $t_{k,0}$ can be the initial points of P2, which will guarantee the feasibility. We notice that the objective function is an affine function and all constraints are convex so it can be solved easily by CVX. To solve P3 efficiently, we design an iterative algorithm shown as Algorithm \ref{Alg1} to solve it iteratively. It is worth to point out that, unlike P2, the initial points $c_k^{(0)}$, $d_k^{(0)}$ and $t_{k,0}^{(0)}$ in P3 can be generated randomly because the feasibility of P3 can be always guaranteed. \par
\begin{algorithm}[tp]
	\caption{The Beamforming Optimization Algorithm }\label{Alg2}
	\begin{algorithmic}[1] 
		\STATE {{\bf Initialize:} fixed feasible points \{$c_k^{*(0)}$, $d_k^{*(0)}$, $t_{k,0}^{*(0)}, $\} $\forall k $, $\epsilon = 0.001$, $m=0$.}
		\WHILE {$\sum\limits_{k=1}^K {\rm Tr}(\mathbf{W}_k^{(m)})-\sum\limits_{k=1}^K {\rm Tr}(\mathbf{W}_k^{(m+1)}) \geq \epsilon$}
		\STATE {Update beamforming matrix \{$\mathbf{W}_k^{(m)}, \alpha_k^{(m)}$\}, $\forall k $ by solving P2 with the fixed feasible point\{$c_k^{*(m)}$, $d_k^{*(m)}$, $t_{k,0}^{*(m)} $\}, $\forall k$}.
		\STATE {Update \{$c_k^{*(m)}$, $d_k^{*(m)}$, $t_{k,0}^{*(m)}$\}, $\forall k $ based on \eqref{eq15}, \eqref{up_d} and \eqref{up_t} respectively.}
		\STATE {$m=m+1$.}
		\ENDWHILE
		\STATE {Update $\alpha_k^* = \alpha_k^{(m)}, \forall k $}
		\STATE {Update beamforming vector $\mathbf{w}_k^{*}, \forall k$ by decomposing $\mathbf{W}_k^{(m)}, \forall k$} based on Gaussian Randomization method.
		\STATE {{\bf Output} \{$\mathbf{w}_k^{*}, \alpha_k^*$\}, $\forall k$}
	\end{algorithmic}\label{Al1}
\end{algorithm}

\begin{algorithm}[tp]
	\caption{The Proposed Alternating Algorithm}\label{Alg3}
	\begin{algorithmic}[1] 
		\STATE {{\bf Initialize:} $\epsilon = 0.001, j=0$.}
		\WHILE { $\sum\limits_{k=1}^K {\rm Tr}(\mathbf{W}_k^{*(j)}) - \sum\limits_{k=1}^K {\rm Tr}(\mathbf{W}_k^{*(j)}) \geq \epsilon$ }
		\STATE {Searching initial fixed feasible point \{$c_k^{*(j)}$, $d_k^{*(j)}$, $t_{k,0}^{*(j)}$\}, $\forall k $} based on Algorithm 1.
		\STATE {Update \{$\mathbf{W}_k^{*(j)}$,$\mathbf{w}_k^{*(j)}$, $\alpha_k^{*(j)}$\}, $\forall k$ based on Algorithm 2.}
		\STATE {Update $\mathbf{V}_k^{*(j)}, \forall k$ by solving P5 with \{$\mathbf{w}_k^{*(j)}$, $\alpha_k^{*(j)}$\}, $\forall k$}
		\STATE {Update phase shift vector $\mathbf{e}_k^{*(j)}, \forall k$ by decomposing $\mathbf{V}_k^{*(j)}, \forall k$ based on Gaussian Randomization method.}
		\STATE {$j = j + 1$}
		\ENDWHILE
		\STATE {{\bf Output} \{$\mathbf{w}_k^{*(j)}$, $\alpha_k^{*(j)}$, $\mathbf{e}_k^{*(j)}$\}, $\forall k$.}
	\end{algorithmic}\label{Al1}
\end{algorithm}

After deciding the fixed points, the last challenge for solving the beamforming optimization problem has been removed. To solve this problem efficiently, we design an alternating algorithm to solve P2 iteratively. The details of the algorithm are shown in Algorithm 2. Specifically, the fixed initial points \{$c_k^{*(0)}$, $d_k^{*(0)}$, $t_{k,0}^{*(0)}$\} $\forall k$ are obtained from Algorithm 1.
\subsection{Phase Shifting Optimization}
In this section, we focus on the phase shifting optimization. The phase shifting optimization can be transformed to a feasibility problem since the objective function in the primal problem does not contain the phase shifting parameter $\mathbf{\Theta}_k, \forall k$. Only the constraints \eqref{P1b}, \eqref{P1c} and \eqref{P1d} in the primal problem contain the phase shifting parameter and \eqref{P1b} can be equivalently divided into \eqref{P2b} and \eqref{P2c}, where only \eqref{P2b} contains the phase shifting parameter. Therefore, given the beamforming vectors, the phase shift feasibility problem can be written as follows:
\begin{subequations}\label{Prob:4} 
\begin{align}
&{\rm P}4:{\rm find} \quad \mathbf{\Theta} \\
		&{\rm s.t.} \quad \log_2(1+{\rm SINR}_{k,e}) \geq R_{k,e}, \forall k \label{P5a}\\
		& \qquad \;\; 0 \leq \theta_{i,n} \leq 2\pi,  \forall i,n \label{P5b}\\
		& \qquad \; |\mathbf{\Theta}_{i,n,n}| = 1, \forall  i,n . \label{P5c}
\end{align}
\end{subequations}
Note that P4 is non-convex due to constraint \eqref{P5a}. It is straighforward to find out that the non-convexity arises form constraint \eqref{P5a}. We need to transform this non-convex constraint to be a convex one. Thus, we can rewrite \eqref{P5a} as follows:
\begin{align}
&|\mathbf{h}_{k,e}^H \Gamma_{\mathbf{p}_k}\mathbf{e}_k|^2(1+r_{k,e})\alpha_k \leq \notag \\
&|\mathbf{h}_{k,e}^H \Gamma_{\mathbf{p}_k}\mathbf{e}_k|^2 - \sum\limits_{\substack{i=1 \\ i \neq k}}^K |\mathbf{h}_{k,e}^H \Gamma_{\mathbf{p}_i}\mathbf{e}_k|^2 - \sigma^2 r_{k,e}, \label{theta1}
\end{align}
where $\Gamma_{\mathbf{p}_i}$ is a diagonal matrix whose main diagonal elements are from $\mathbf{p}_i = \mathbf{G}_k \mathbf{w}_i$ and $e_k$ is the phase shifting vector. However, with ${\mathbf{W}_k}, \alpha_k, \forall k$ already obtained from the beamforming optimization problem, constraint \eqref{theta1} is a quartic form with respect to $\mathbf{e}_k$. For simplicity, we substitute $\mathbf{h}_{k,e}^H \Gamma_{\mathbf{p}_i}$ with $\mathbf{r}_{k,e}^{iH}$. From \cite{luo2010semidefinite}, we know that a quartic form can be equivalently transformed to a linear form with a rank-one constraint. Thus, \eqref{theta1} can be expressed as follows:
\begin{align}
&{\rm Tr}(\mathbf{R}_{k,e}^k \mathbf{V}_k)(1+r_{k,e})\alpha_k \leq \notag \\
&{\rm Tr}(\mathbf{R}_{k,e}^k \mathbf{V}_k)- \sum\limits_{\substack{i=1 \\ i \neq k}}^K {\rm Tr}(\mathbf{R}_{k,e}^i \mathbf{V}_i)- \sigma^2 r_{k,e} \label{theta2}\\
&\mathbf{V}_k \succcurlyeq 0\label{theta3}\\
&{\rm Rank}(\mathbf{V}_k) = 1, \label{thetha4}
\end{align}
where $\mathbf{R}_{k,e}^i = \mathbf{r}_{k,e}^i \mathbf{r}_{k,e}^{iH}$ and $\mathbf{V}_i = \mathbf{e}_i \mathbf{e}_i^H$. Given $\mathbf{w}_k, \alpha_k, \forall k$, \eqref{theta2} is an affine constraint. The rank-one constraint will make the whole problem intractable, so we adopt SDR to remove this rank-one constraint frist. Then P4 can be transformed as follows:
\begin{subequations}\label{Prob:5} 
\begin{align}
&{\rm P}5:{\rm find} \quad \mathbf{V}_k, \quad \forall k \\
		&{\rm s.t.} \qquad \; \eqref{theta2}, \quad \forall k  \label{P6a} \\
		& \qquad \; \mathbf{V}_k \succcurlyeq 0, \quad \forall k \label{P6b}\\
		& \qquad \; \mathbf{V}_{k,n,n} = 1, \quad \forall k,n. \label{P6c}
\end{align}
\end{subequations}
P5 is a convex problem, which can be solved by CVX efficiently. Since the rank-one constraint is removed, the optimal solution of P5 may not be the optimal solution of P4. Therefore, Gaussian randomization will be applied to achieve a sub-optimal solution for P4.
\subsection{Algorithm Design}
The details of the proposed alternating algorithm are illustrated in Algorithm 3, where P2 and P5 are alternately solved until the convergence metric is satisfied. At the $i$-th iteration of Algorithm 3, first, the initial points are obtained by Algorithm 1. Then, the algorithm begins to solve the beamforming optimization problem by solving P2 through Algorithm 2. Then, the algorithm starts to solve phase shifting feasibility problem by solving P5 (step 5 and step 6) to obtain a feasible phase shift vector $\mathbf{e}_k^{*(i)}, \forall k$. The feasible phase shifting vector of this current iteration will be used as a given phase shift for the beamforming optimization in the next iteration. It is worth to point out that after each iteration, the channel state will change with the new obtained phase shifting vector $\mathbf{e}_k, \forall k$, so in each iteration before solving P2, we need to search new feasible fixed points (step 3), which is necessarily to guarantee that P2 is feasible at each iteration.
\subsection{Complexity analysis}
From \cite{luo2010semidefinite}, we learn that the worst complexity of solving a SDR problem through CVX is 
\begin{equation} \label{cvxcom}
\mathcal{O}({\rm max} \{m,n\}^4 n^{1/2} \log(1/\epsilon_{c})),
\end{equation}
where $n$ is the problem size, and $m$ is the number of constraints and $\epsilon_{c}$ is the accuracy of the algorithm that CVX adopts. We assume that the problem size is greater than the number of constraints, then we have the complexity of CVX to solve a SDR problem as 
\begin{equation} \label{cvxcom2}
\mathcal{O}(n^{4.5} \log(1/\epsilon_{c})).
\end{equation}
Algorithm 1 is essentially to solve a SDR problem multiple times until the accuracy is satisfied. Thus, the complexity of Algorithm 1 is
\begin{equation}
\mathcal{O}\left(n_1^{4.5} \log\left(\frac{1}{\epsilon_{c}}\right) \log\left( \frac{1}{\epsilon_{1}}\right)\right),
\end{equation}
where $n_1$ is the problem size of P3 and $\epsilon_1$ is the accuracy of Algorithm 1. Algorithm 2 is similar to Algorithm 1, which is also to solve a SDR problem multiple times and hence P2 has the same size as P3. Thus, the complexity of Algorithm 2 can be expressed as follows:
\begin{equation}
\mathcal{O}\left(n_1^{4.5} \log\left(\frac{1}{\epsilon_{c}}\right) \log\left(\frac{1}{\epsilon_{2}}\right)\right),
\end{equation}
where $\epsilon_2$ is the accuracy of Algorithm 2. Now, we have the complexities of step 3 and step 4 in Algorithm 3. The last one we need to consider is the complexity of step 5. It is easy to find out that a single SDR  problem is solved in the step 5, so the complexity is
\begin{equation}
\mathcal{O}(n_2^{4.5} \log(1/\epsilon_{c}),
\end{equation}
where $n_2$ is the problem size of P5. Finally, we have the complexity of the proposed algorithm as follows:
\begin{equation} \label{complprop}
\mathcal{O} (\mathcal{O}_1\log(1/\epsilon_3)),
\end{equation}
where 
\begin{align}
\mathcal{O}_1 = \; &n_1^{4.5} \left(\log\left(\frac{1}{\epsilon_{c}}\right) \log\left(\frac{1}{\epsilon_{1}}\right) + \log\left(\frac{1}{\epsilon_{c}}\right) \log\left(\frac{1}{\epsilon_{2}}\right) \right) + \notag \\
&n_2^{4.5} \log(1/\epsilon_{c})\notag.
\end{align}
\section{Partial Exhaustive Search Algorithm}
In this section, we propose a simple algorithm based on partial exhaustive search, which can reduce significantly computation complexity. The main idea of this partial exhaustive search algorithm is to assume that all the clusters using the same power allocation coefficient, of which the optimal value can be obtained by an exhaustive search within the range [0, 1]. The primal problem can also be divided into the beamforming optimization problem and the phase shifting feasibility problem.

Since each cluster shares the same power coefficient, the power coefficient is first fixed in each searching progress so we only need to optimize the beamforming vector and the phase shifting vector in these two subproblems. We notice that these two subproblems can be reduced to the QCQP problem, which is a classic form in convex optimization theory. SDR is widely used as one of the most common methods to efficiently solve the QCQP problem. Two subprobelms are formulated as P6 and P7. We can obtain P6 and P7 through the basic SDR theory and some simple algebraic transformations, where the derivation is omitted in this paper due to space limitations. 
\begin{subequations}\label{Prob:6}
\begin{align}
&{\rm P6}:\quad \min\limits_{\mathbf{\mathbf{w}}} \quad \sum\limits_{k=1}^K {\rm Tr}(\mathbf{W}_k) \\
&{\rm s.t.} \; \alpha {\rm Tr}(\mathbf{H}_{k,c} \mathbf{W}_k) \geq \sum\limits_{\substack{i=1 \\ i \neq k}}^K {\rm Tr}(\mathbf{H}_{k,c} \mathbf{W}_i) r_{k,c} + \sigma^2 r_{k,c}, \forall k \label{P7a} \\
& \alpha {\rm Tr}(\mathbf{Z}_{k,e} \mathbf{W}_k) \leq \notag \\
& \frac{{\rm Tr}(\mathbf{Z}_{k,e}\mathbf{W}_k)}{1+r_{k,e}} - (\sum\limits_{\substack{i=1 \\ i\neq k}}^K {\rm Tr}(\mathbf{Z}_{k,e}\mathbf{W}_i) + \sigma^2) \frac{r_{k,e}}{1+r_{k.e}} , \forall k \label{P7b} \\
& \alpha {\rm Tr}(\mathbf{H}_{k,c} \mathbf{W}_k) \leq \notag \\
& \frac{{\rm Tr}(\mathbf{Z}_{k,e}\mathbf{W}_k)}{1+r_{k,e}} - (\sum\limits_{\substack{i=1 \\ i\neq k}}^K {\rm Tr}(\mathbf{H}_{k,c}\mathbf{W}_i) + \sigma^2) \frac{r_{k,e}}{1+r_{k.e}}, \forall k \label{P7c} \\
& \qquad \qquad \qquad \mathbf{W}_k \succcurlyeq 0, \forall k \label{P7d} 
\end{align}
\end{subequations}

where $\mathbf{Z}_{k,e}, \mathbf{W}_k, \forall k$ in P6 are the same as those in P2.
\begin{subequations}\label{Prob:7}
\begin{align}
&{\rm P7}:{\rm find} \quad \mathbf{V}_k, \quad \forall k \\
&{\rm s.t.} \quad {\rm Tr}(\mathbf{R}_{k,e}^k \mathbf{V}_k)(1+r_{k,e})\alpha \leq \notag \\
&\qquad \; {\rm Tr}(\mathbf{R}_{k,e}^k \mathbf{V}_k)- \sum\limits_{\substack{i=1 \\ i \neq k}}^K {\rm Tr}(\mathbf{R}_{k,e}^i \mathbf{V}_i)- \sigma^2 r_{k,e}, \forall k \label{P8a}\\
&\qquad \; \mathbf{V}_k \succcurlyeq 0, \quad \forall k \label{P8b}\\
& \qquad \; \mathbf{V}_{k,n,n} = 1, \quad \forall k,n \label{P8c}
\end{align}
\end{subequations}
where $\mathbf{R}_{k,e}^i, \forall i,k$ and $\mathbf{V}_k, \forall k$ are the same as those in P5. The detail of the partial exhaustive search algorithm is illustrated in Algorithm \ref{Alg4}. \par
In each search progress, the algorithm will solve two SDR problems with different sizes $n_1$ and $n_2$, which are the problem sizes of P6 and P7. Therefore, the complexity of Algorithm \ref{Alg4} can be expressed as follows:
\begin{equation}\label{complxqiong}
\mathcal{O}\left(I \left( n_1^{4.5} \log(1/\epsilon_{c})  + n_2^{4.5} \log(1/\epsilon_{c}) \right)\right).
\end{equation}
$I$ is the number of searches, which is depended on the search step $\alpha$. Comparing \eqref{complprop} and \eqref{complxqiong}, we can find that the partial exhaustive search algorithm has a lower complexity than the one in the previous section.

\begin{algorithm}[tp]
	\caption{The Partial Exhaustive Search Algorithm}\label{Alg4}
	\begin{algorithmic}[1]
		\STATE{{\bf Initialization} $P_{opt} = 10000, \alpha_{opt} = 0, \mathbf{w}_k^*,  \mathbf{e}_k^*, \forall k$}
		\FOR{ $\alpha = 0.1:0.1:0.9$ }
		\STATE {{\bf Initialization} $\epsilon = 0.001, i = 0, \mathbf{e}_k^{(0)}$}
		\WHILE {$\sum\limits_{k=1}^K {\rm Tr}(\mathbf{W}_k^{(i)}) - \sum\limits_{k=1}^K {\rm Tr}(\mathbf{W}_k^{(i+1)}) > \epsilon$}
		\STATE{Update $\mathbf{W}_k^{(i)}, \forall k$ by solving P6}.
		\STATE{Update $\mathbf{w}_k^{(i)}, \forall k$ by decomposing $\mathbf{W}_k^{(i)}, \forall k$ based on Gaussian Randomization method.}
		\STATE {Update $\mathbf{V}_k^{(i)}, \forall k$ by solving P7 based on given $\mathbf{w}_k^{(i)}, \forall k$.}
		\STATE {Update $\mathbf{e}_k^{(i)}, \forall k$ by decomposing $\mathbf{V}_k^{(i)}, \forall k$ based on Gaussian Randomization method.}
		\STATE{$i = i+1.$}
		\ENDWHILE
		\IF{$P_{opt} > \sum\limits_{k=1}^K {\rm Tr}(\mathbf{W}_k^{(i)})$}
		\STATE{$P_{opt} > \sum\limits_{k=1}^K {\rm Tr}(\mathbf{W}_k^{(i)})$.}
		\STATE{$\alpha_{opt} = \alpha$, $\mathbf{w}_k^* = \mathbf{w}_k^{(i)}$, $\mathbf{e}_k^* = \mathbf{e}_k^{(i)}$, $\forall k$.}
		\ENDIF
		\ENDFOR
		\STATE{{\bf Output} $\alpha_{opt}, \mathbf{w}_k^*, \mathbf{e}_k^*, \forall k$.}
	\end{algorithmic}\label{Al1}
\end{algorithm}
\section{Numerical Results}
In this section, we evaluate all simulation results of the proposed algorithms. In simulations, channel gains are generated by
\begin{align}
\mathbf{h}_{k,e} = \frac{\mathbf{h}_{k,e}^*}{\sqrt{ d_0^{\alpha_0}}} \quad 
\mathbf{G}_k = \frac{\mathbf{G}_k^*}{\sqrt{d_1^{\alpha_1}}} \quad 
\mathbf{h}_{k,c} = \frac{\mathbf{h}_{k,c}^*}{\sqrt{ d_2^{\alpha_2}}} 
\end{align}
where $k = 1,2,...K$, and  $\mathbf{h}_{k,e}^*$, $\mathbf{G}_k^*$, $\mathbf{h}_{k,c}^*$ are complex Reyleigh channel coefficients. $d_0 = 10$ m, $d_1 = 50$ m, $d_2 = 10$ m, respectively denote the distances between the IRS and the cell edge user, the distance between the BS and the IRS, and the distance between the BS and the cell center user. $\alpha_0 ,\alpha_1 ,\alpha_2$ are the path loss exponents of the corresponding links. We assume that all the cell central users are at the same distance from the BS, all the cell edge users are at the same distance from the related IRS and all the IRSs are at the same distance from the BS. We set $\alpha_0 = \alpha_2 = 2$ and $\alpha_1 = 2.2$. The noise power is $\sigma^2 = BN_0$, where the bandwidth $B = 100$ MHz and the noise power spectral density is $N_0 = -80$ dBm. \par

Fig. \ref{Fig1} shows the transmit power at the BS versus the number of each IRS's reflecting elements. We provide the performance of the proposed schemes compared with the random phase scheme in NOMA and OMA. In Fig. \ref{Fig1}, the number of antennas at the BS is $M = 8$, and the date rate requirement of all the central users and the cell edge users is 1 bps/Hz. Obviously, the transmit power at the BS of all schemes decreases rapidly with the increasing of the number of IRS's reflecting elements. From Fig. \ref{Fig1}, we can see that both proposed algorithms requires a less transmit power than the benchmarks. Comparing the two proposed algorithms, the performance gap is very small at the beginning but is gradually enlarging when the number of reflecting elements at IRSs grows. The result in Fig. \ref{Fig1} demonstrates that the alternating algorithm can yield the best performance among the schemes shown in the figure. \par

Fig. \ref{Fig2} shows the transmit power at the BS versus the minimum data rate of the central users. In this figure, we assume that the central users' date rate requirement are the same, and all the cell edge users' date rates are set as 1bps/Hz. In this figure, we set the number of antennas at the BS as $M = 8$ and the number of reflecting elements at each IRS as $N = 32$ respectively. According to the Shannon's capacity formula, it is well known that a higher date rate requires a higher transmit power at the BS. All schemes in Fig. \ref{Fig2} have the same trend, where the transmit power at the BS increases with the increasing of the central users' minimum date rate. From Fig. \ref{Fig2}, we find that the proposed alternating algorithm needs less power consumption under the same date rate requirement. Although, the partial exhaustive search algorithm cannot achieve the same performance as the proposed alternating algorithm, it has low complexity and still yields a better performance than NOMA with random IRS scheme and OMA scheme. \par

\begin{figure}
	\centering
	\graphicspath{{./figures/}}\includegraphics[width=1.1\linewidth]{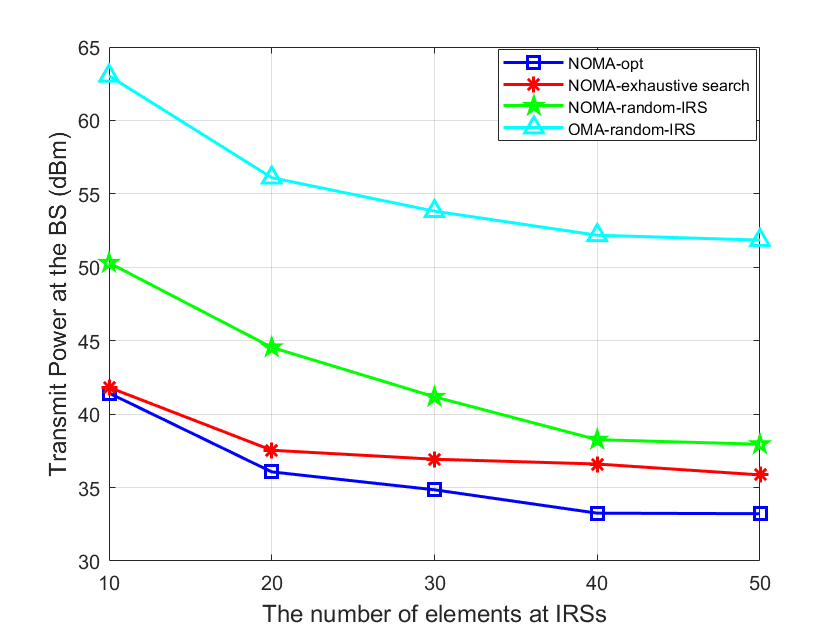}\\
	\caption{The transmit power versus the number of elements.} \label{Fig1}
\end{figure}

\begin{figure}
	\centering
	\graphicspath{{./figures/}}\includegraphics[width=1.1\linewidth]{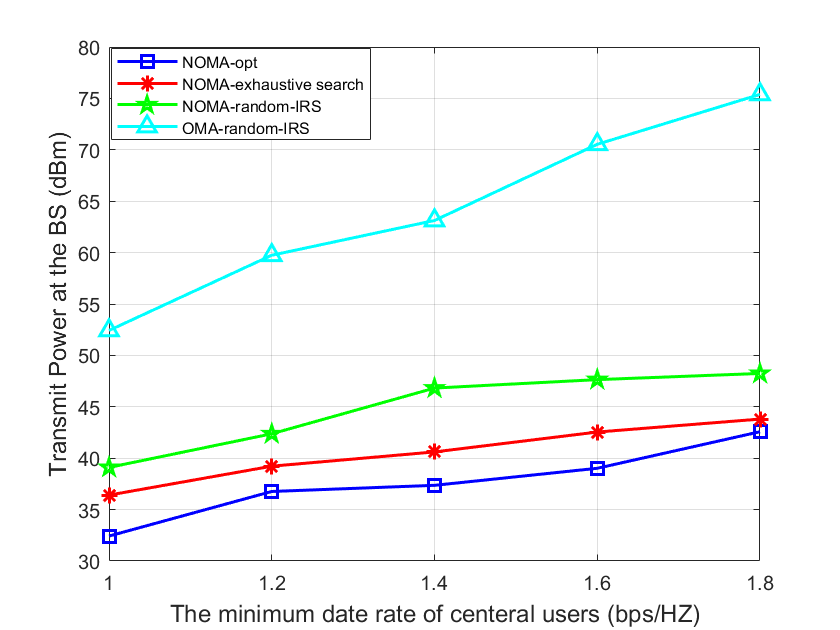}\\
	\caption{The transmit power versus the minimum date rate of the central users.} \label{Fig2}
\end{figure}

\begin{figure}
	\centering
	\graphicspath{{./figures/}}\includegraphics[width=1.1\linewidth]{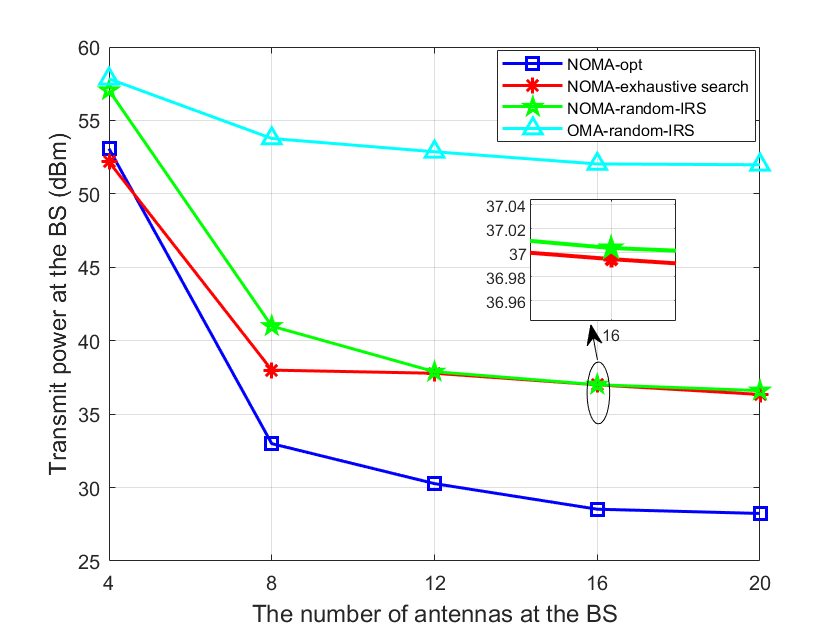}\\
	\caption{The transmit power versus the number of antennas at the BS.} \label{Fig3}
\end{figure}

\begin{figure}
	\centering
	\graphicspath{{./figures/}}\includegraphics[width=1.1\linewidth]{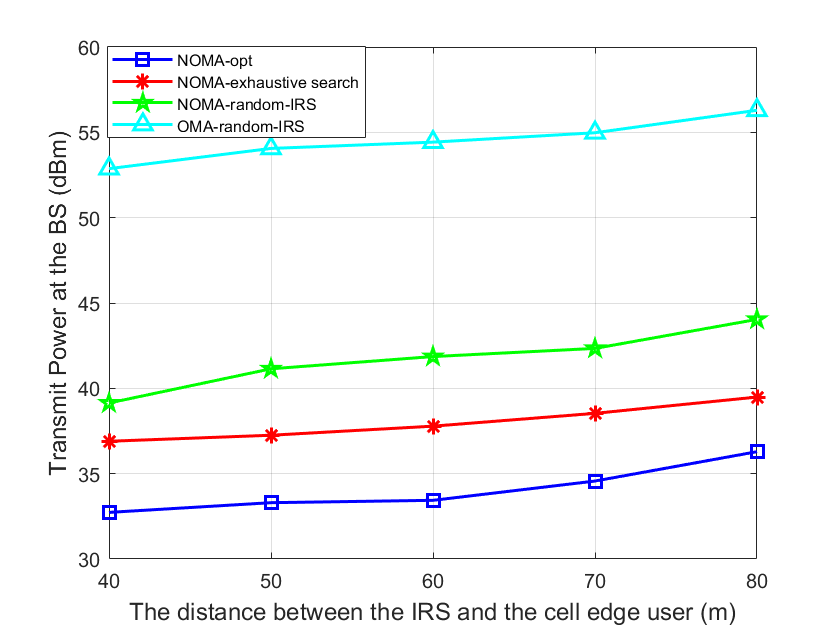}\\
	\caption{The transmit power versus the distance between the IRS and the cell edge user.} \label{Fig4}
\end{figure}

\begin{figure}
	\centering
	\graphicspath{{./figures/}}\includegraphics[width=1.1\linewidth]{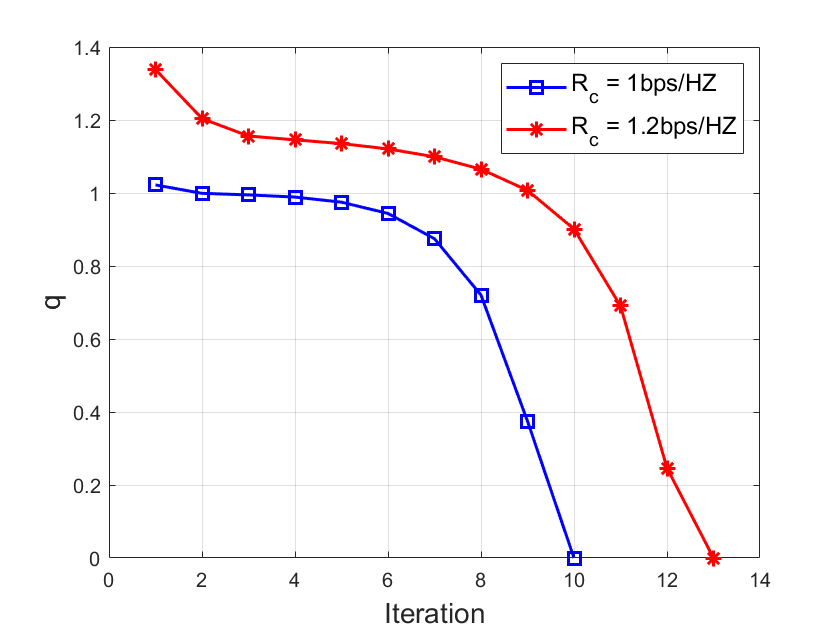}\\
	\caption{The value of $q$ versus the iterative number.} \label{Fig5}
\end{figure}

\begin{figure}
	\centering
	\graphicspath{{./figures/}}\includegraphics[width=1.1\linewidth]{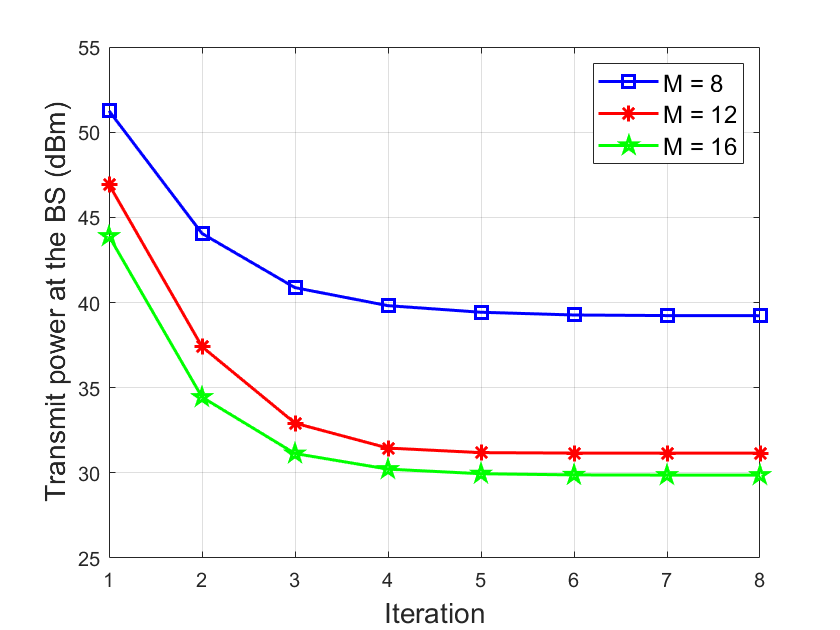}\\
	\caption{The transmit power at the BS versus the iterative number.} \label{Fig6}
\end{figure}
Fig. \ref{Fig3} shows the transmit power versus the number of antennas at the BS. Similar to Fig. \ref{Fig1}, we also compare the performance of two proposed algorithms and NOMA and OMA with random phase IRS. In Fig. \ref{Fig3}, we set the number of reflecting elements at each IRS as $N = 32$ and the date rate requirements of all users are chosen as 1 bps/Hz. Obviously, we can see the transmit power at the BS of all schemes decreases with an increasing number of antennas at the BS, and these two proposed algorithms are still better than the benchmarks. When the number of antennas at the BS is set as 4, the performance of the partial exhaustive search algorithm is slightly better than the alternating algorithm. However, the energy consumption of the proposed alternating algorithm drops faster than the partial exhaustive search algorithm with the number of base station's antennas increasing. It indicates that the proposed alternating algorithm performs better than the partial exhaustive search method in the most of cases of massive antennas. \par

Fig. \ref{Fig4} shows the transmit power versus the distance between the IRS and the cell edge user in each cluster. In Fig. \ref{Fig4}, we set the number of antennas at the BS as $M = 8$, the number of each IRS's reflecting elements as $N = 32$. Each users' minimum date rate is 1bps/Hz. As expected, the transmit power of all schemes increases when the distance between IRS and the cell edge user in each cluster gets large. Similar to Fig. \ref{Fig3}, the proposed alternating algorithm consumes less energy compared with all other schemes. \par

Fig. \ref{Fig5} shows the the value of $q$ in the initial point search algorithm versus the iterative number. As we discussed before, the $q$ represents the distance between the current problem and a feasible problem and $q$ can enforce the current problem to be a feasible one. In Fig. \ref{Fig4}, we can see that the value of $q$ in the $R_c = 1.2$ bps/Hz scheme is larger than that in the $R_c = 1$ bps/Hz scheme at each iteration. Moreover, the scheme with $R_c = 1.2$ bps/Hz needs more iterations to converge, which indicates that a higher date rate makes all constraints more difficult to be fulfilled. \par

Fig. \ref{Fig6} shows the transmit power at the BS versus the iterative number in Algorithm 2. We evaluate the transmit power in different scenarios with the different numbers of antennas at the BS. We set all users' data rate as 1 bps/Hz and the number of each IRS as $N = 32$. From Fig. \ref{Fig6}, we find that the transmit power at the BS decreases with the number of iterations increasing, which also means this algorithm can converge with the algorithm proceeding. 

\section{Conclusion}
In this paper, we investigated joint optimization of beamforiming, power allocation and IRS phase shift in a NOMA-IRS assisted multi-cluster network. By introducing inequality approximation, SCA and SDR, we proposed an alternating algorithm to minimize the transmit power by iteratively solving beamforming optimization and phase shifting feasibility until the algorithm converges. Furthermore, we proposed an initial point search algorithm to guarantee the feasibility of the beamforming optimization subproblem. Moreover, to reduce the complexity of the proposed algorithm, we also provided a low-complexity solution for this scenario based on the partial exhaustive search. The simulation results demonstrated the alternating algorithm outperforms the simplified partial exhaustive search algorithm but has a higher complexity. In the future research, the IRS reconfigures the imperfect channel will be studied and the inter-cluster interference caused by IRS will also be considered. 

\bibliographystyle{IEEEtran}

\vspace{0.5em}
\end{document}